\documentclass[a4paper,11pt,openright]{article} 
\usepackage[english]{babel} 
\usepackage{amsmath,amssymb,amsthm,amsfonts,amscd,authblk,dsfont,color,cancel,diagbox,enumerate,epsfig,eucal,fancyhdr,latexsym,lmodern,mathrsfs,mathtools,multicol,multirow,musicography,phaistos,skull,simpler-wick,xcolor,subcaption,tikz-cd}
\usepackage[normalem]{ulem} 
\usepackage[multiuser]{fixme}
\RequirePackage[latin1,utf8]{inputenc}
\FXRegisterAuthor{nd}{annd}{Nico} 
\FXRegisterAuthor{cvdv}{ancvdv}{Chris} 
\usetikzlibrary{shapes,snakes} 

\usepackage[
final=true,                    
plainpages=false,              
pdfpagelabels=true,            
pdfencoding=auto,              
unicode=true,                  
hypertexnames=true,            
naturalnames=true              
]{hyperref}
\usepackage[all,cmtip]{xy}
\usetikzlibrary{matrix,calc}
\usepackage[autostyle, english = american]{csquotes}
\MakeOuterQuote{"} 

\definecolor{NiColor}{RGB}{77,77,255}
\definecolor{NiColoRed}{RGB}{255,77,77}
\definecolor{NiCitation}{RGB}{0,181,26}

\newtheoremstyle{TheoremStyle}
{3pt}
{3pt}
{\slshape}
{}
{\sc}
{:}
{.5em}
{}

\theoremstyle{TheoremStyle}
\newtheorem{theorem}{Theorem}

\newtheorem{proposition}[theorem]{Proposition}
\newtheorem{lemma}[theorem]{Lemma}
\theoremstyle{definition}
\newtheorem{definition}[theorem]{Definition}
\newtheorem*{remark}{Remark}
\newtheorem{example}[theorem]{Example} 
\makeatletter
\def\@endtheorem{\hfill$\lozenge$} 
\makeatother

\title{Large deviations in mean-field quantum spin systems}
\author[a]{\href{mailto:matthias.keller@uni-potsdam.de}{Matthias Keller}}
\author[b]{\href{mailto:christiaan.van-de-ven@uni-tuebingen.de}{Christiaan J. F. van de Ven}}
\affil[a]{University of Potsdam, Department of Mathematics, Campus Golm, Haus 9 Karl-Liebknecht-Stra\ss e 24-25, 14476 Potsdam, Germany}
\affil[b]{University of T\"{u}bingen, Department of Mathematics, Auf der Morgenstelle 10, 72076  T\"{u}bingen, Germany}

\def\bearray{\begin{eqnarray}}
\def\earray{\end{eqnarray}}
\def\beq{\begin{equation}}
\def\eeq{\end{equation}}

\def\b0{{\bf 0}}

\newsymbol\bt 1202             

\def\bC{{\mathbb C}}           

\newsymbol\rest 1316         
\def\gA{{\mathfrak A}}       

\newcommand{\Hmm}[1]{\leavevmode{\marginpar{\tiny%
$\hbox to 0mm{\hspace*{-0.5mm}$\leftarrow$\hss}%
\vcenter{\vrule depth 0.1mm height 0.1mm width \the\marginparwidth}%
\hbox to 0mm{\hss$\rightarrow$\hspace*{-0.5mm}}$\\\relax\raggedright #1}}}
 
\begin{document}


\maketitle

\begin{abstract}
\noindent
Continuous fields (or bundles) of $C^*$-algebras form an important ingredient for describing emergent phenomena, such as phase transitions and spontaneous symmetry breaking. 
In this work, we consider the continuous $C^*$-bundle generated by increasing symmetric tensor powers of the complex $\ell\times\ell$ matrices $M_\ell(\mathbb{C})$, which can be interpreted as abstract description of mean-field theories defining the macroscopic limit of infinite quantum systems. Within this framework  we discuss the principle of large deviations for the local Gibbs state in the high temperature regime and characterize the limit of the ensuing logarithmic generating function. 
\end{abstract}

\tableofcontents

\section{Introduction}

\subsection{Continuous $C^*$-bundles and $C^*$-algebraic deformation quantization}\label{subs:intro}
In this paper two important and seemingly unrelated topics in mathematical physics are addressed and brought forward by relating them to each other:
\begin{itemize}
\item[1.] A continuous bundle of $C^*$-algebras, with additional structure of deformation quantization;
\item[2.] The theory of large deviations.
\end{itemize}
The relationship is that the first topic may be used to describe a class of ``physical'' Hamiltonians, for which one is allowed to study large deviations of suitable observables by using the continuity properties of the $C^*$-bundle, yielding a classical (commutative) theory in the pertinent limit. More precisely, for the purpose of this work, we adapt this setting to mean-field quantum spin Hamiltonians.
\\\\
Below is the technical definition of continuous bundle of $C^*$-algebras, \textit{cf.} \cite[\S IV.1.6]{Blackadar_2006}. We denote by $\overline{\mathbb{N}}=\mathbb{N}\cup\{\infty\}$, and by $C(\overline{\mathbb{N}})$ the space of $\mathbb{C}$-valued sequences $(\alpha_N)_{N\in\mathbb{N}}$ such that $\alpha_\infty:=\lim\limits_{N\to\infty}\alpha_N\in\mathbb{C}$ exists.
A \textbf{continuous bundle (or field) of $C^*$-algebras} over $\overline{\mathbb{N}}$ is a triple $\mathcal{A}$, $(\mathcal{A}_N)_{N\in\overline{\mathbb{N}}}$, $(\chi_N)_{N\in\overline{\mathbb{N}}}$ made by $C^*$-algebras $\mathcal{A},\mathcal{A}_N$, $N\in\overline{\mathbb{N}}$, and surjective homomorphisms $\chi_N\colon\mathcal{A}\to\mathcal{A}_N$ such that:
\begin{enumerate}[(i)]
	\item
	\label{Item: norm condition for bundle of Cstar algebras}
	The norm of $\mathcal{A}$ is given by $\|a\|_{\mathcal{A}}=\sup\limits_{N\in\overline{\mathbb{N}}}\|\chi_N(a)\|_{\mathcal{A}_N}$;
	
	\item
	\label{Item: product by function condition for bundle of Cstar algebras}
	For all $\alpha=(\alpha_N)_{N\in\overline{\mathbb{N}}}\in C(\overline{\mathbb{N}})$ and $a\in\mathcal{A}$ there exists  $\alpha a\in\mathcal{A}$ with the property that $\chi_N(\alpha a)=\alpha_N\chi_N(a)$;
	
	\item
	\label{Item: norm continuity for bundle of Cstar algebras}
	For all $a\in\mathcal{A}$, $(\|\chi_N(a)\|_{\mathcal{A}_N})_{N\in\overline{\mathbb{N}}}\in C(\overline{\mathbb{N}})$.
\end{enumerate}
A \textbf{continuous section} of $\mathcal{A}$ is an element $a\in\prod_{N\in\overline{\mathbb{N}}}\mathcal{A}_N$ such that there exists $a'\in\mathcal{A}$ fulfilling $a_N=\chi_N(a')$ for all $N\in\overline{\mathbb{N}}$. In this way $\mathcal{A}$ can be identified with its continuous sections. Indeed, if we define the map
\begin{align*}
\mathcal{A}\to \prod_{N\in\overline{\mathbb{N}}}\mathcal{A}_N; 
 \ \ \ \ \ a\mapsto (\chi_N(a))_{N\in\overline{\mathbb{N}}},
\end{align*}
it follows by surjectivity of $\chi_N$ and (i) that this map is a bijection. 
For this reason, we will implicitly identify $\chi_N$, $N\in\overline{\mathbb{N}}$, with the projection $\prod_{N\in\overline{\mathbb{N}}}\mathcal{A}_N\to\mathcal{A}_N$.
\\\\
The above construction may lead to a \textbf{strict deformation quantization}, defined by the following additional data, see e.g.  \cite{Landsman_2017,DV2}:
\begin{enumerate}
	\item
	A continuous bundle of $C^*$-algebras, as described above for which the fiber at infinity assumes the form $\mathcal{A}_\infty=C(X)$, where $X$ is a Poisson manifold;

	\item  
	A family of linear maps, called \textbf{quantization maps}, $Q_N\colon\widetilde{\mathcal{A}}_\infty\to\mathcal{A}_N$, $N\in\overline{\mathbb{N}}$, with $\widetilde{\mathcal{A}}_\infty:=C^\infty(X)$, such that
	\begin{enumerate}
		\item\label{Item: quantization maps are Hermitian and define a continuous section}
		$Q_\infty=\operatorname{Id}_{\widetilde{\mathcal{A}}_\infty}$ and $Q_N(a_\infty)^*=Q_N(a_\infty^*)$ for all $a_\infty\in\widetilde{\mathcal{A}}_\infty$ and all $N\in\overline{\mathbb{N}}$.
		Moreover, the assignment
		\begin{align}\label{ctssection}
			\overline{\mathbb{N}}\ni N\mapsto Q_N(a_\infty)\in\mathcal{A}_N\,,
		\end{align}
		defines a continuous section of the bundle.
		\item\label{Item: quantization maps fulfils the DGR condition}
		For all $a_\infty,a_\infty'\in\widetilde{\mathcal{A}}_\infty$ one has the {\bf Dirac-Groenewold-Rieffel condition}
		\begin{align}\label{Eq: Dirac-Groenewold-Rieffel condition}
			\lim_{N\to\infty}\|Q_N(\{a_\infty,a_\infty'\})-iN[Q_N(a_\infty),Q_N(a_\infty')]\|_{\mathcal{A}_N}=0\,.
		\end{align}
		
		\item\label{Item: quantization maps are strict}
		For all $N\in\mathbb{N}$, $Q_N(\widetilde{\mathcal{A}}_\infty)$ is a dense $*$-subalgebra of $\mathcal{A}_N$.
	\end{enumerate}
\end{enumerate}
The algebra $\mathcal{A}_\infty$ at $N=\infty$ represents the classical (or macroscopic) observables of the physical system, whilst the fibers $\mathcal{A}_N$, $N\in\mathbb{N}$, recollect the quantum observables of the increasing quantum system. The continuous sections defined by the quantization maps, \textit{cf.} \eqref{Item: quantization maps are Hermitian and define a continuous section} of the definition above, form the main ingredients of our analysis. In this paper the continuous $C^*$-bundle of our interest is equipped with the additional structure of a strict deformation quantization, \textit{cf.} $\S\ref{MFQss and all that}$.

We would like to point out to the reader that this abstract framework is perfectly suited to characterize the notion of the {\em classical limit} of quantum systems, including a rigorous derivation of spontaneous symmetry breaking and phase transitions \cite{MorVen1,MorVen2,LMV,Ven20222,Ven2022}. 

\subsection{Principle of large deviations}
Large deviations play an important role in the asymptotic evaluation of certain integrals often appearing in probability theory and classical statistical mechanics, viz.
\begin{align}\label{integral0}
\int d\mu_N(x)e^{\nu_N f(x)},\ \  (\nu_N\to\infty, \ \text{as} \ N\to\infty)
\end{align}
for which the measures $\mu_N$ satisfy a law of large numbers, $\nu_N$ are real positive numbers and $f$ is some function.  Traditionally, the Laplace method, and more generally, the method of the steepest descent, have turned out to be important approximation methods for such integrals, in particular in the theory of soliton equations, integrable models and random matrices. It is however a challenging and difficult task to make these techniques mathematically precise. Instead, the principle of large deviations, due to Varadhan \cite{Varadhan,Varadhan0}, provides an efficient way of evaluating  integrals of the above type. The abstract setting is that the measures $\mu_N$ are probabilities measures  on the Borel subsets of a Polish space $\Omega$, i.e., a separable completely metrizable topological space.
Then, in certain situations, one may expect that
$$d\mu_N(x)\ \sim \  dx e^{-\nu_N I(x)},$$ 
in a suitable sense, where $I$ is the so-called {\em rate function} for the sequence of measures $(\mu_N)_N$. In such cases, the integral \eqref{integral0} may be computed as
\begin{align*}
\frac{1}{\nu_N}\log{\int d\mu_N(x) e^{\nu_N f(x)}}\ \to \ \sup{\{f(x)-I(x) \ |\ x\in \Omega\}}, \ \  (N\to\infty).
\end{align*}
To make all this precise the following well-known definition is introduced.
\begin{definition}[LDP]\label{LDP I}
A sequence of probability measures $(\mu_N)_{N\in\mathbb{N}}$ on a Polish space
$\Omega$ equipped with the Borel $\Sigma$-algebra  satisfies a {\bf large deviation principle} (LDP) with a \textbf{rate function} $I : \Omega \to [0, \infty]$ and speed (or rate) $\nu_N\geq 0$, if
\begin{itemize}
    \item[(i)] $I$ has closed sub-level sets $\{x\in \Omega \ | \ I(x)\leq k\}$ for all $k\in [0,\infty)$.
    \item[(ii)] For all compact $C\subset \Omega$,
\begin{align*}
\limsup_{N\to\infty}\frac{1}{\nu_N}\log{\mu_N(C)}\leq -\inf_{x\in C}I(x).
\end{align*}
\item[(iii)] For all open $O\subset \Omega$,
\begin{align*}
\liminf_{N\to\infty}\frac{1}{\nu_N}\log{\mu_N(O)}\geq -\inf_{x\in O}I(x).
\end{align*}
\end{itemize}  Furthermore, $I$ is said to be a \textbf{good} rate function if all the level sets are compact subsets of $\Omega$.
\hfill$\blacksquare$
\end{definition}
The theory of large deviations can be applied in  several different fields, in particular in semi-classical analysis of Schr\"{o}dinger operators in the regime $\hbar\to 0$, where $\hbar$ appears in front of the Laplacian. In this way, up to a suitable scale separation, $\nu_N$ should be interpreted as $1/\hbar$. This has proved to be a rigorous approach for the study of quantum tunneling and eigenfuntion localization \cite{Sim851,Sim852}. The rate function in such  cases is a quantitative measure of exponential decay away from the minima of the pertinent potential, as function of $\hbar$.
\\\\
In order to study the large deviations for $\mu_N$ as $N\to\infty$, one typically considers the corresponding logarithmic moment generating
function, defined as
\begin{align*}
F(t):=\lim_{N\to\infty}\frac{1}{\nu_N}\log{\int_\Omega d\mu_N(x)e^{Ntx}}.
\end{align*} 
The G\"{a}rtner–Ellis Theorem (see e.g. \cite{EG}), shows that the existence of
$F(t)$ implies large deviation upper bounds with good rate function $I(x)$,  that is, the Legendre transform of $F(t)$. One obtains lower bounds if, in addition, the moment generating function is differentiable. If the moment generating function is not differentiable, one has a weaker result: in Definition~\ref{LDP I}~(iii), the
infimum over $O$ is replaced by the infimum over $O\cap E$, where $E$ is the set  of the so-called ``exposed points'' \cite{Dembo}.
\subsection{Large deviations in statistical mechanics}\label{ldp in statmech}
Beside its great value in classical statistical mechanics, the theory of large deviations is also suitable for probability measures occurring in (equilibrium) quantum statistical mechanics. In this setting, one typically seeks for a large deviation principle concerning the celebrated Kubo–Martin–Schwinger equilibrium states (KMS states), which are, loosely speaking, limit points of the local Gibbs state induced by local interacting Hamiltonians. For each such  KMS state $\omega$  a particular study of interest is to analyze the asymptotic behaviour of the distribution of  quantum averages with respect to $\omega$, in the regime of increasing number of lattice sites or increasing volume \cite{LR,NR,O}. 

In this work we focus on mean-field quantum spin systems, arising as abstract elements of a continuous bundle of $C^*$-algebras \cite{Ven2020,Ven2022}. Due to the long range interacting nature of such models, no KMS condition can be formulated for the limiting system. Nonetheless, on account of the renown Quantum De Finetti Theorem, one may rely on a suitable decomposition of the limit points of local Gibbs states to investigate large deviations. The minimizers of the pertinent mean-field free energy, whose existence is always guaranteed \cite{CLR,RMNB,Ven20222}, are uniquely identified with the limit points of the local Gibbs state \cite{GRV}. Each minimizer represents a pure thermodynamic phase of the given infinite quantum system, whose thermodynamics is described by the fixed phase space $S(M_\ell(\mathbb{C}))$, interpreted as the set of density matrices of size $\ell$. The ensuing mean-field limit is therefore considered as ``macroscopic'': the number of particles is appropriately  sent to infinity  yielding a commutative limiting theory as a limit of infinite statistical mechanics.
\\\\
\noindent
This paper is structured as follows.  The principle of large deviations is first investigated in the $C^*$-algebraic framework of deformation quantization and continuous $C^*$-bundles, namely the one describing quantum spin systems with mean-field interactions on a lattice of increasing size, also dubbed {\em symmetric $C^*$-bundle}.  More precisely, in Theorem~\ref{prop: moment generating function 21} we prove large deviation upper bounds, and discuss the lower bounds, for the logarithmic moment generating function, defined in what follows. For $(h_N)_N$, an arbitrary self-adjoint continuous cross-sections of the symmetric $C^*$-bundle of matrix algebras (the only case we care about), we first consider the {\bf local Gibbs state}
$$\omega_N(\cdot):=\frac{Tr[e^{Nh_N}(\cdot)]}{Tr[e^{Nh_N}]},$$
see \cite{BR12} for further context and details.
For $(a_N)_N$ another arbitrary self-adjoint continuous cross-sections of this symmetric $C^*$-bundle, we then define $F_N:\mathbb{R}\to \mathbb{R}$
\begin{align}\label{momgen1}
F_N(t):=\frac{1}{N}\log{\frac{Tr[e^{Nh_N}e^{Nta_N}]}{Tr[e^{Nh_N}]}}=\frac{1}{N}\log{\omega_N(e^{Nta_N})}.
\end{align}
Note that $(h_N)_N$ and $(a_N)_N$ are generalizations of the empirical averages, the latter often used in this context \cite{LR,NR,O}.  Here, $N$ corresponds to the number of lattice sites  of the  underlying quantum system; and therefore plays the role of the diverging sequence $(\nu_N)_N$ in the above discussion. The section $(h_N)_N$ defines the normalized version of a mean-field quantum spin Hamiltonian, \textit{cf.} \eqref{meanfield},\eqref{asymptotic equivalence 2} below.
%
%

More precisely,  if  $(a_N)_N$  is such a cross-section, it is known that the spectrum $\sigma(a_N)$ of $a_N$ is contained in $\text{ran}(a_\infty)$, where $a_\infty$ denotes the principal symbol of the section  $(a_N)_N$, that is, a real-valued continuous function seen as limit of the sequence $(a_N)_N$ \cite{Ven20222}. One may then consider the spectral projection $1_\Delta(a_N)$ defined for Borel sets $\Delta\subset\text{ran}(a_\infty)$, and define the {\em distribution of $a$} with respect to the local Gibbs state $\omega_N$
\begin{align}\label{mu}
\mu_N(\Delta):=\omega_N(1_\Delta(a_N)).
\end{align}



Consequently, we deduce large deviation upper bounds, and discuss the lower bounds, for the sequence of probability measures $(\mu_N)_{N\in\mathbb N}$. We hereto demonstrate  that the generating function $F_N$, which, by functional calculus also reads
\begin{align*}\label{generating function}
F_N(t)=\frac{1}{N}\log{\int e^{tNx}d\mu_N(x)}
\end{align*}
has a limit as $N\to\infty$.

\begin{remark}
Our result actually relates to the celebrated Born rule: if an observable $a$ is measured in a state $\omega$, then the probability $\mathbb{P}_\omega (a \in \Delta)$ that the
outcome lies in some measurable subset $\Delta\subset\sigma(a)\subset\mathbb{R}$ is given by
\begin{align*}
    \mathbb{P}_\omega (a \in \Delta)=\mu_\omega(\Delta).\tag*{$\blacksquare$}
\end{align*}
\end{remark}
\begin{remark}
One should be aware of the fact that Golden-Thompson inequality, stating that 
\begin{align}\label{GT}
    \frac{1}{N}\log{Tr[e^{Nh_N+Nta_N}]}\leq \frac{1}{N}\log{Tr[e^{Nh_N}e^{Nta_N}]}, \ \ (N\in\mathbb{N}),
\end{align}
remains in general a strict inequality in the limit, even in the case that $\lim_{N\to\infty}\|[h_N,a_N]\|_N= 0$, cf. \cite{RW}.
    \hfill$\blacksquare$
\end{remark}

Finally, we would like to point out to the reader that the possible existence of the limit of  $F_N$ defined by \eqref{momgen1} for mean-field models $(h_N)_N$ and $(a_N)_N$ has been conjectured in \cite{RW}, no proof was however given in subsequent literature. 



\section{Mean-field quantum spin systems in algebraic quantum theory}\label{LDP mean field theories}
\subsection{Mean-field theories}\label{mfts}
Mean-field theories (MFTs) play an important role as approximate models of the more complex nearest neighbor interacting spin systems. Their relatively simple structure allows for a detailed analysis of the limit of increasing number of lattice sites especially in view of spontaneous symmetry breaking (SSB) and phase transitions \cite{Ven20222}. 

Homogeneous mean-field quantum spin systems fall into the class of MFTs. They are defined by a single-site Hilbert space $\mathbb{C}^\ell$ and local Hamiltonians of the type
\begin{align}
H_\Lambda=|\Lambda|h_\infty(t_0^{(\Lambda)},t_1^{(\Lambda)},\cdot\cdot\cdot, t_{\ell^2-1}^{(\Lambda)}),\label{meanfield}
\end{align}
where $h_\infty$ is the {\em principal symbol}, i.e., a polynomial in $\ell^2$ (non-commutative) variables, and $\Lambda\subset\mathbb{Z}^d$ denotes a finite lattice on which $H_{\Lambda}$ is defined and $|\Lambda|$ denotes the number of lattice points (see e.g. \cite[$\S$10.8]{Landsman_2017}). Here $t_0= 1_{\ell}$ is the identity matrix in  the  algebra of complex $\ell\times\ell$ matrices $M_\ell(\mathbb{C})$, and the matrices $(t_\mu)_{\mu=1}^{\ell^2-1}$ in $M_\ell(\mathbb{C})$ form a basis of the real vector space of traceless self-adjoint $\ell\times \ell$ matrices; the latter may be identified with $i$ times the Lie algebra $\mathfrak{su}(\ell)$ of $SU(\ell)$, which identification is useful in defining the Poisson bracket below, {cf.} \eqref{Poisson} below. 
The macroscopic average spin operators are now defined by
\begin{align}\label{basisliealgebra}
t_\mu^{(\Lambda)}=\frac{1}{|\Lambda|}\sum_{x\in\Lambda}t_\mu(x), \ \ (\mu=1,\cdots ,\ell^2-1).
\end{align}
Here $t_\mu(x)$ stands for $1_{\ell} \otimes \cdots \otimes t_\mu\otimes \cdots \otimes 1_{\ell}$, where $t_\mu$ occupies  slot $x$.

It follows directly from the definition that such models are characterized by the property that all spins interact with each other which implies that these models are permutation-invariant and that the geometric configuration including the dimension is irrelevant. In what follows we therefore consider homogeneous mean-field quantum spin chains, i.e. $d=1$. We will leave out the term ``homogeneous'' in the forthcoming  discussion.

\subsection{Strict deformation quantization of Poisson manifolds}\label{MFQss and all that}
In contrast to the {\em thermodynamic limit}, that is, a rigorous formalism in which the number of lattice sites $N$ as well as the volume of the system at constant density, are sent to infinity, the limit we consider here is {\em macroscopic}.
This means that the limit $N\to\infty$ is formalized by a classical (commutative) theory, introduced in what follows.
\\

\noindent For any unital $C^*$-algebra $\mathfrak{A}$ and $X\subset\mathbb{Z}$ with $N=|X|$, we set
\begin{align}\label{Eq: quasi-local bundle}
    \mathcal{A}_{N}^\pi:=\begin{dcases}
        \gA_X^\pi\equiv \pi_N(\mathfrak{A}^{N})\quad  N\in\mathbb{N},
        \\
        [\mathfrak{A}]_\infty^\pi\quad \hspace*{1.65cm}N=\infty,
    \end{dcases}
\end{align}
where $\pi_N$ is the \textbf{symmetrization operator} defined by continuous and  linear extension on elementary tensors $a_N=\alpha_1\otimes\cdots\otimes \alpha_N$
\begin{align}\label{def:SN}
\pi_N(a_N):=
    \frac{1}{N!}\sum_{\sigma\in S_N}\alpha_{\sigma(1)}\otimes\cdots\otimes \alpha_{\sigma(N)}\,
\end{align}
where $\alpha_i\in\mathfrak{A}$ for all $i=1\dots N$,
and the summation is over the elements $\sigma$ in the permutation group of order $N!$, denoted by $S_N$. Moreover, the algebra $[\mathfrak{A}]_\pi^\infty$ is defined in an inductive way, yielding a so-called generalized inductive system \cite[$\S$V.4.3]{Blackadar_2006}. This algebra is actually the completion of equivalence classes of symmetric sequences (introduced below) under the equivalence relation
\begin{align*}
a_N \sim b_N \  \ \ \text{iff}  \  \ \  \lim_{N\to\infty}\|a_N-b_N\|_N=0,
\end{align*}
and norm
\begin{align*}
\| [a_N]_N\|=\lim_{N\to\infty}\|a_N\|_N.
\end{align*}
This construction leads to a continuous field of $C^*$-algebras whose continuous sections are given by    quasi-symmetric sequences \cite[Theorem 8.4]{Landsman_2017}. To construct these we need to generalize the definition of $\pi_N$. For $\Lambda\subset X$  with $M=|\Lambda|$, define a bounded  operator  $\pi_N^M: {\mathfrak{A}^\Lambda}	\to {\mathfrak{A}}^X$, defined by linear and continuous extension of  
\begin{align}\pi_N^M(a) = \pi_N(a\otimes \underbrace{1 \otimes \cdots \otimes 1 }_{N-M \mbox{\scriptsize \ times}}),\quad a \in {\mathfrak{A}}^\Lambda. \label{defSMN}
\end{align}
Clearly, $\pi_N^N=\pi_N$. We say that a sequence
 $(a_{N})_{N\in\mathbb{N}}$ is called (strictly) {\bf symmetric}
if there exist $\Lambda \subset X$ (with $M=|\Lambda|$) and $a_\Lambda\in {\mathfrak{A}}^{\Lambda}$ such that 
\begin{align*}
    a_N= \pi_N^M(a_\Lambda)\ \ \:\mbox{for all }  N\geq M,
\end{align*}
and {\bf quasi-symmetric} if  $a_N = \pi_N(a_N)$ if $N\in \mathbb{N}$,
and for every $\varepsilon > 0$, there is a symmetric sequence $(a_N')_{N\in\mathbb{N}}$
as well as  $\Lambda \subset X$ (both depending on $\varepsilon$) such that 
\begin{align*}
 \|a_{N}-a_{N}'\| < \varepsilon\: \mbox{ for all } N > M.
\end{align*} 
\begin{example}\label{Example}
If $\lim_{N\to\infty}\|a_N-a_N'\|_N= 0$ for some fixed symmetric sequence $(a_N')$, then $(a_N)$ is obviously quasi-symmetric.
\hfill$\blacksquare$
\end{example}
\begin{remark}
From a classical point of view, quasi-symmetric sequences are also called ``tail events'', i.e., events  whose occurrence is not altered by local changes.
\hfill $\blacksquare$
\end{remark}
We write $S({\mathfrak{A}})$ for the algebraic state space, consisting of all positive linear normalized functionals $\omega:\mathfrak{A}\to\mathbb{C}$, and indicate by $C(S(\mathfrak{A}))$ the $C^*$-algebra of continuous functions on $S(\mathfrak{A})$.  It is common knowledge and not hard to see (see \cite[Chapter 8]{Landsman_2017}, \cite{LMV} for details), that for any quasi-symmetric sequence the following limit exists
\begin{align}
a_\infty(\omega)=\lim_{N\to\infty}\omega^N(a_{N}),\label{8.46}
\end{align}
where $\omega\in S({\mathfrak{A}})$, and the product state $$\omega^N=  \underbrace{\omega\otimes \cdots \otimes \omega}_{N\: \mbox{\scriptsize times}} \in S(\pi_N({\mathfrak{A}}^{N}))$$
is  the unique (norm) continuous linear extension of the following map that is defined on elementary tensors
\begin{align*}
\omega^N(\alpha_1\otimes\cdot\cdot\cdot\otimes \alpha_N)=\omega(\alpha_1)\cdot\cdot\cdot\omega(\alpha_N).
\end{align*}
Furthermore, the limit in \eqref{8.46} defines a function in $C(S(\mathfrak{A}))$ provided that $(a_{N})_{N\in\mathbb{N}}$ is quasi-symmetric, otherwise it may not exist. In particular, the continuous cross-sections are  the sequences of the form $(a_N)_{N\in\overline{\mathbb{N}}}=(a_N,a_\infty)_{N\in \mathbb{N}}$, where $a_\infty$ is defined by \eqref{8.46}. These sections are also called {\em macroscopic}, due to their commutative character in the limit $N\to\infty$. 

In fact, an application of the Quantum De Finetti Theorem entails that $[\mathfrak{A}]_\infty^\pi$ is commutative and isometrically isomorphic to $C(S(\mathfrak{A}))$, \cite[Chapter~8]{Landsman_2017}, \cite{RW}. This isomorphism is precisely implemented by the following map, \textit{cf.} \eqref{8.46}
\begin{align}\label{map j}
j:[\mathfrak{A}]_\infty^\pi\to C(S(\mathfrak{A})),\qquad j([a_N]_N)(\omega):=\lim_{N\to\infty}\omega^N(a_N)=a_\infty(\omega). 
\end{align}
While the above construction works for general $\mathfrak{A}$, more can be said in the specific case that $\mathfrak{A}=M_\ell(\mathbb{C})$. Indeed,  the above construction then relates to a strict deformation quantization of the state space $S(M_\ell(\mathbb{C}))$ \cite[Theorem~3.4]{LMV}. Specifically, there exists an affine parametrization of $S(M_\ell(\mathbb{C}))$ representing $S(M_\ell(\mathbb{C}))$ as a compact convex subset $\mathcal{Q}_\ell\subset\mathbb{R}^{\ell^2-1}$ 
with non-empty interior  which coordinates are explicitly given by
\begin{align}\label{coordinates} S(M_\ell(\mathbb{C}))\to \mathcal{Q}_\ell,\quad \omega\mapsto (x_\mu(\omega))_\mu^{\ell^2-1}= (\omega_{t_\mu})_\mu^{\ell^2-1},
\end{align}
where the  matrices $(t_\mu)_{\mu=1}^{\ell^2-1}$ in $M_\ell(\mathbb{C})$ are a normalized basis of the real vector space of traceless self-adjoint $\ell\times \ell$ matrices, see  \cite[Eqn. (2.11)]{LMV}.

Moreover, the interior $\text{int}(\mathcal{Q}_\ell)$ of $\mathcal{Q}_\ell$ corresponds to the
rank-$\ell$ density matrices and is a connected $\ell^2-1$ dimensional smooth manifold \cite[$\S$2.2]{LMV}. In this realization, $S(M_\ell(\mathbb{C}))$ is shown to admit a canonical Poisson structure, given for $f,g\in C^\infty(\mathcal{Q}_\ell)$ by
\begin{align}\label{Poisson}
\{f,g\}(x)=\sum_{a,b,c=1}^{\ell^2-1}C_{ab}^cx_c\frac{\partial \tilde f(x)}{\partial x_a}\frac{\partial \tilde g(x)}{\partial x_b}, \ \ x\in C^\infty(\mathcal{Q}_\ell),
\end{align}
where $\tilde f,\tilde g\in C^\infty(\mathbb{R}^{\ell^2-1})$ are arbitrary extensions of $f,g$ and $C_{ab}^c$ are the structure constants corresponding to the lie algebra of $SU(\ell)$. The precise technical details of this construction can be found in \cite[$\S$2.2, $\S$2.3]{LMV}.
\\\\
The quantization maps $Q_{N}$ are defined on a dense Poisson algebra $\Tilde{\mathcal{A}}_\infty^\pi\subset\mathcal{A}_\infty^\pi\simeq C(S(M_\ell(\mathbb{C})))$ as follows. First, the relevant Poisson subalgebra $\Tilde{\mathcal{A}}_\infty^\pi$ is made of the restrictions to $S(M_\ell(\mathbb{C}))$ of polynomials in $\ell^2-1$ coordinates of $\mathbb{R}^{\ell^2-1}$. Each elementary  tensor of the form $$t_{j_1}\otimes\cdot\cdot\cdot\otimes t_{j_L},$$
where $it_1,\ldots, it_{\ell^2-1}$ are traceless self-adjoint $\ell\times\ell$ matrices forming a basis of the Lie algebra of $SU(\ell)$, may be uniquely identified with a monomial $p_L$ of degree $L$ \cite[Lemma 3.2]{LMV}. This allows one to define the quantization map $Q_{N}$ as follows. If 
\begin{align*}
    p_L(x_1,\ldots, x_{\ell^2-1}) = x_{j_1} \cdots x_{j_L}\quad \mbox{where $j_1,\ldots, j_L \in   \{1,2,\ldots, \ell^2-1\}$,}
\end{align*}
the quantization maps $Q_{N}: \Tilde{\mathcal{A}}_\infty^\pi \to \mathcal{A}_X^\pi$ act as 
\begin{align}\label{deformationquantization1}
 Q_{N}(p_L) &=
\begin{cases}
    \pi_N^L(t_{j_1}\otimes\cdot\cdot\cdot\otimes t_{j_L}), &\ \text{if} \ N\geq L \\
    0, & \ \text{if} \ N< L,
\end{cases}\\
Q_{N}(1) &= \underbrace{1_\ell \otimes \cdots \otimes 1_\ell}_{\scriptsize N \: \mbox{times}}, \label{deformationquantization2}
\end{align}
where $\pi_N^L$ is given by \eqref{defSMN}. We extend $Q_N$ to all polynomials by taking linear combinations of the elementary tensors $t_{j_1}\otimes\cdot\cdot\cdot\otimes t_{j_L}$.
The  maps $Q_{N}$ indeed satisfy the required conditions \eqref{Item: quantization maps are Hermitian and define a continuous section},\eqref{Item: quantization maps fulfils the DGR condition},\eqref{Item: quantization maps are strict} of the pertinent definition displayed in $\S$\ref{subs:intro}, cf. \cite{LMV}, whose properties will be used in $\S$\ref{App:proofs}.
\begin{remark}\label{correspondence rmk}
It is precisely this interpretation of the limit $N\to\infty$ that relates mean field quantum spin systems to strict deformation quantization, since the  ensuing mean-field Hamiltonians correspond to quasi-symmetric sequences which in turn are defined by the quantization maps \eqref{deformationquantization1}--\eqref{deformationquantization2}. A straightforward combinatorial exercise shows that \eqref{meanfield} implies
 \begin{align}\label{asymptotic equivalence 2}
\|H_N/N-Q_N(h_\infty)\|_N=O(1/N),
\end{align}
where $h_\infty$  is a polynomial (principal symbol) defining the mean-field model. 
We point out that the scaling factor $N$ is essential: it says that the mean-field Hamiltonian is of order $N$, which resembles the physical idea that the spectral radius scales with the dimension of the underlying space.
\hfill$\blacksquare$
\end{remark}

\section{Logarithmic moment generating function and local Gibbs states}\label{magicmoments}
From now on $\gA=M_\ell(\mathbb{C})$ and we consider the continuous field of $C^*$-algebras with continuous cross-sections given by quasi-symmetric sequences.  We first discuss the principle of large deviations for local Gibbs states, in that, we focus on the quantity $F_N$ defined by \eqref{momgen1}.
Our main result is Theorem~\ref{prop: moment generating function 21}, proving sufficient conditions for the limit $F(t):=\lim_{N\to\infty}F_N(t)$ to exist for all small $t$.  \\

\noindent To this avail, we proceed as follows. We rewrite
\begin{align*}
    F_N(t)=\frac{1}{N}\log{Tr[e^{N\bar Z_{N,t}}]}-\frac{1}{N}\log{Tr[e^{N h_N}]}
\end{align*}
where
\begin{align}\label{section}
    \bar Z_{N,t}=\frac{1}{N}\log{e^{Nta_N}e^{Nh_N}},
\end{align}
which is well-defined since $e^{Na_N}e^{Nb_N}$ is invertible. We will prove that $[\bar Z_{N,t}]_N\in[\mathfrak{A}]_\infty^\pi$. Subsequently, we show that $\lim_{N\to\infty}F_N(t)$ exists and characterize the limit.

\subsection{A Baker-Campbell-Hausdorff approach}\label{App:proofs}

For any matrices $a,b\in M_\ell(\mathbb{C})$ and $s\in\mathbb{R}$, we consider the following formal expression for the Baker-Campbell-Hausdorff (BCH) formula
\begin{align} \label{eq 28}
e^{sa}e^{sb}=e^{Z(s)}, \ \  Z(s):=\log{(e^{sa}e^{sb})}=\sum_{n=1}^\infty s^nz(n),
\end{align}
where the Taylor coefficients are given as usual by
\begin{align*}
    z(n)=\frac{1}{n!}\left.\frac{\partial^n}{\partial s^n}Z(s)\right\vert_{s=0}.
\end{align*}
Note that the first terms can be computed explicitly, they are given by  the nested commutators
\begin{align*}z(1)=a+b, \ \ z(2)=\frac{1}{2}[a,b], \ \ z(3)=\frac{1}{6}[z(2),b-a],\ \ z(4)=\frac{1}{12}[[z(3),a],b].
\end{align*}
In order to analyze all the coefficients, we follow the well known strategy  found in the book of Varadarajan \cite{Varadarajan}.
It is a well known fact that for $s$ sufficiently small the series $\sum_{n=1}^\infty s^nz(n)$ converges absolutely \cite{Mueger,Suzuki}. 
Furthermore, for all such small $s$, $Z(s)$ is a solution to  the differential equation \cite[Lemma 2.15.2]{Varadarajan}
\begin{align*}
    \frac{dZ}{ds}=a+b+\frac{1}{2}[a-b,Z]+\sum_{p=1}^\infty\frac{B_{2p}}{(2p)!}\text{ad}_Z^{2p}(a+b),
\end{align*}
with initial condition $Z(0)=0$. Here $\text{ad}_Z^0(a)=a$, $\text{ad}_Z^k(a)=[Z,{\text{ad}_Z}^{k-1}(a)]$ and $B_{2p}$ are the Bernoulli numbers.
This allows one to determine {\em all} the coefficients $z(n)$, as done in e.g. \cite[Lemma 2.15.3]{Varadarajan}. Indeed,  they can be shown to be uniquely determined by the following recursion formulae
\begin{align}
&z(1)=a+b;\label{expansion}\\
&z(n)=\frac{1}{2n}[a-b,z(n-1)]+\frac{1}{n}\hspace{-.2cm}\sum_{\substack{p\geq 1\\ 2p\leq n-1}}\hspace{-.3cm}\frac{B_{2p}}{(2p)!}\hspace{-.2cm}\sum_{\substack{k_1,\cdots ,k_{2p}>0\\ k_1+\cdots k_{2p}=n-1}}\hspace{-.6cm}[z(k_1),[\cdots [z(k_{2p}),a+b]{\scriptsize\cdots} ]\nonumber\\&=: z(a,b,[\cdot,\cdot],n). \label{expansion2}
\end{align}
Observe that we include the commutator in the definition of $ z(a,b,[\cdot,\cdot],n)$ since we will replace it later also with a Poisson bracket $\{\cdot,\cdot\}$.

\subsection{Method of majorants}\label{Methodofmaj}
There exist several approaches proving the convergence of the series $Z(1)$ \cite{Mueger,Suzuki}. For reasons that will become clear in the next paragraph, we follow the {\em method of majorants} originating  with Varadarajan  \cite[Theorem~2.15.4]{Varadarajan}, briefly outlined in what follows. 
First, as seen above the coefficients $z(n)$ are determined by a recursion scheme involving only rational coefficients (defining the $z(n)$). One can set up a parallel recursion scheme where these coefficients are replaced by their absolute values. The solution to this parallel problem is always positive and can be shown to define a convergent solution, while the original solution is majorized by the solution to the parallel problem. This trick of ``comparing'' solutions  will therefore prove  convergence of the original problem. 
\\\\
\noindent
To this avail, consider the differential equation
\begin{align}
    &\frac{dy}{dz}=\frac{1}{2}y+\sum_{p=1}^\infty\frac{|B_{2p}|}{(2p)!}y^{2p},\label{ode}\\
    &y(0)=0.\nonumber
\end{align}
It can be shown that there is a uniform constant  $\delta>0$ such that the above equation admits a solution $y$ which is holomorphic in the disc $\{z \mid |z| <\delta \}$, \cite{Varadarajan}. The coefficients appearing in the ensuing power series defining $y$ satisfy a similar recurrence relation as the $z(n)$, \textit{cf.} \eqref{expansion2}.
The key idea is that the $z(n)$ can be majorized by the coefficients defining $y$ up to an exponential factor. This allows one to prove convergence of the BCH series $Z(1)$. The precise result is stated below, further details regarding the method of majorants can be found in the book by Varadarajan \cite{Varadarajan}.
\begin{theorem}[Theorem~2.15.4 in \cite{Varadarajan}]\label{varr}
Let $\delta>0$ be given as above. The series $Z(1)=\sum_{n=1}^\infty z(n)$ converges absolutely for all $a,b \in \mathfrak{L}:=\{c \in M_\ell(\mathbb{C})\ | \ \|c\|< \delta/4\}$, its sum defines an analytic map from $\mathfrak{L}\times\mathfrak{L}$ into $M_\ell(\mathbb{C})$, and $e^{a}e^{b}=e^{Z}$. 
\end{theorem}

\subsection{A ``classical'' counterpart of the BCH formula}
To see how this adapts to our case, we take a finite set $X\subset\mathbb{Z}$,  $N=|X|$, two fixed positive natural numbers $M,M'$ and consider $$a_N=\pi_N^M(a_M)\quad \mbox{and}\quad b_N=\pi_N^{M'}(b_{M'}) $$ two strict symmetric sequences, i.e., continuous cross-sections of the symmetric $C^*$-bundle. We are interested in the sequence $\bar Z_{N}\equiv \bar{Z}_{N,1}$ given by \eqref{section}, that is, the element $Z(s)$ defined by \eqref{eq 28} for $s=1$ with $a$ and $b$ from above replaced by $Na_N$ and $Nb_N$, and scaled by an overall factor $1/N$. According to the analysis above $\bar Z_N$ assumes the following formal form
\begin{align}\label{barZ}
\bar Z_N=\frac{1}{N}\sum_{n=1}^\infty N^nz_N(n)=\sum_{n=1}^\infty \bar z_N(n),
\end{align}
where $$z_N(n)=z(Na_N,Nb_N,[\cdot,\cdot],n)\quad\mbox{ and }\quad
 \bar z_N(n)=  \frac{1}{N}z_N(n)$$
i.e., we obtain a similar expression as \eqref{expansion2}
$$ \bar z_N(1)=a_N+b_N,\; \bar z_N(2)=\frac{N}{2}[a_N,b_N],\;\bar{z}_N(3)=\frac{N}{6}[\bar z_N(2),b_N-a_N],\,\ldots.$$
\noindent
We do not claim that the formal sum \eqref{barZ} converges for finite $N$.

First of all, from the  definition of the $\bar z_N(n)$ it is not at all evident how to bound the norms of $a_N$ and $b_N$ such that $Z_N$ indeed converges. The standard estimate $\|[a_N,b_N]\|_N\leq 2\|a_N\|\|b_N\|$ does not work because of the pre-factors $N$ causing divergent behavior. 

A more sophisticated bound would be to use the Dirac-Groenewold-Rieffel condition, \textit{cf.} \eqref{Eq: Dirac-Groenewold-Rieffel condition} allowing to write each $\bar{z}_N(n)$ as a quantized  $n$-fold Poisson bracket up to order $O(1/N)$. However, it is not true that the $O(1/N)$ error terms hold independently of $n$. It is therefore not at all clear whether the existence of the series \eqref{barZ} holds for any finite $N$, which makes  a usual analysis as done in e.g. \cite{Mueger,Suzuki} not possible. 

Nonetheless, at the level of equivalence classes we show below that the series indeed exists and defines an element of $[\mathfrak{A}]_\infty^\pi$.
\\\\
\noindent
To this end, we focus on the method of majorants, introduced in $\S \ref{Methodofmaj}$.
It is precisely these techniques that will be adapted to prove convergence in our case. In fact, we first prove the existence of a ``classical'' analog of the BCH formula, and then use an algebraic argument to obtain convergence of the function $F_N(t)$, as $N\to\infty$ at least for small $t$. To this end, we proceed as follows.
\\\\
\noindent
We define the truncation of $\bar Z_N$ given by \eqref{barZ} as
\begin{align}
    \bar{Z}_N^k:=\sum_{n=1}^kN^{n-1}z_N(n)=\sum_{n=1}^k\bar{z}_N(n).
\end{align}
The next lemma shows that the truncated series is a continuous cross-section.
\begin{lemma}\label{Lemm: fundamental1}
Let $\Lambda,\Lambda'\subset X$ with $M=|\Lambda|$, $M'=|\Lambda'|$ and $N=|X|$.  Consider $a'\in \mathfrak{A}^\Lambda$ and $b'\in\mathfrak{A}^{\Lambda'}$. For $N\geq M,M'$ associate the symmetric sequences $a_N=\pi_N^M(a')$ and  $b_N=\pi_N^{M'}(b')$, as defined $\S \ref{MFQss and all that}$. For any fixed $k>0$, the section defined by 
\begin{align*}
A_N:=\begin{cases}
    \bar{Z}_{N}^k, &N\in\mathbb{N};\\
    [\bar{Z}_{N}^k]_N, &N=\infty
\end{cases}
\end{align*}
is continuous and quasi-symmetric. In particular,  there exists $\bar{Z}_\infty^k\in C(S(\gA))$ such that
 \begin{align}\label{finite k}
 \lim_{N\to\infty}\|Q_N(\bar{Z}_\infty^k)-\bar{Z}_N^k\|_N=0,
 \end{align}
and $\bar{Z}_\infty^k$ is a linear combination of $k$ terms of $n$-fold nested Poisson brackets, for $n\leq k$.
 \end{lemma}
\begin{proof}
On account of the results in \cite[$\S$3]{LMV}, the sequences $(a_N)_N$ and $(b_N)_N$ are symmetric, therefore they arise as quantized polynomials $a_\infty$ and $b_\infty$, i.e. $a_N=Q_N(a_\infty)$ and $b_N=Q_N(b_\infty)$. The polynomials $a_\infty$ and $b_\infty$ are just the classical limits of the ensuing continuous cross-sections, \textit{cf.} \eqref{8.46}.

For $n=1$, the sequence $\bar{z}_N(1)$ satisfies $\bar{z}_N(1)=Q_N(a_\infty)+Q_N(b_\infty)=Q_N(a_\infty+b_\infty)$, by linearity. On account of \eqref{ctssection}, the ensuing cross-section is continuous with limit $\bar{z}_\infty(1)=a_\infty+b_\infty$. It is easy seen that the same holds for  linear combinations of arbitrary (quasi) symmetric continuous cross-sections.

For $n=2$, the Dirac-Groenewold-Rieffel condition \eqref{Eq: Dirac-Groenewold-Rieffel condition}, yields that $\bar{z}_N(2)=N[a_N,b_N]$ is asymptotically norm-equivalent to the sequence defined by $Q_N(\{a_\infty,b_\infty\})$. By \eqref{ctssection}, this is  continuous, so that  $\bar{z}_N(2)$ is continuous itself, with limit $\bar{z}_\infty(2)=\{a_\infty,b_\infty\}$. In fact, $\bar z_N(2)$ induces a quasi-symmetric sequence, as also seen from Example \ref{Example}. The same holds for commutators of arbitrary continuous cross sections.

For $n$ finite but arbitrary, $\bar{z}_N(n)$ is a linear combination of nested commutators of continuous cross sections. The classical limit is given by \eqref{expansion}--\eqref{expansion2} replacing the  commutators (for $\bar{z}_N(n))$  with the Poisson brackets.

As this remains true for (finite) sums, we conclude that $(\bar Z_N^k)_N$ is continuous, with limit $\bar Z_\infty^k=\sum_{n=1}^k\bar{z}_\infty(n)\in C(S(\mathfrak{A}))$. Since $C(S(\mathfrak{A}))\cong [\mathfrak{A}]_\infty^\pi$, \textit{cf.} \eqref{map j}, the ensuing cross-section $(A_N)_{N\in\overline{\mathbb{N}}}$ is continuous and quasi-symmetric.  By construction, \eqref{finite k} holds.
 \end{proof}
 
As already seen in the proof of Lemma \ref{Lemm: fundamental1} one can introduce a classical analog $\bar Z_{\infty}$ of  $\bar{Z}_N$ as follows. Let us denote by  $a_\infty$ and $b_\infty$, the classical limits of the sections $(a_N)_N$ and $(b_N)_N$, the latter being defined as $a_\infty(\omega):=\lim_{N\to\infty}\omega^N(a_N)$,  and similarly for $b_\infty$, \textit{cf.} \eqref{8.46} and \eqref{map j}. The series $\bar{Z}_\infty$ is now  defined by substituting $N$ times the commutator by the Poisson bracket yielding a series defined by the coefficients given by the classical limit of 
\eqref{expansion2}
\begin{align}
    \bar z_\infty(n)= z (a_\infty,b_\infty, \{\cdot,\cdot\} ,n),
\end{align}
i.e., 
\begin{align*}
   \bar{z}_\infty(1)=a_\infty+b_\infty,\;
    \bar{z}_\infty(2)=\frac{1}{2}\{a_\infty,b_\infty\},\;
    \bar{z}_\infty(3)=\frac{1}{6}\{\bar{z}_\infty(2),b_\infty- a_\infty\},\;\ldots
\end{align*}
and so on. In other words, the coefficients $\bar{z}_{\infty}(n) $ are defined by the recurrence relation \eqref{expansion}--\eqref{expansion2}, with $[\cdot,\cdot]$ replaced by $\{\cdot,\cdot\}$.  Note that on account of the proof of Lemma \ref{Lemm: fundamental1} and the estimate
\begin{align*}
    |\omega^N(\bar{z}_N(n))-\omega^N(Q_N(\bar{z}_\infty(n)))|&\leq \|\bar{z}_N(n)-Q_N(\bar{z}_\infty(n))\|_N,
\end{align*}
each $\bar{z}_\infty(n)$ is the classical limit of the $\bar{z}_N(n)$ appearing in the Baker-Campbell-Hausdorff expansion.

Due to the Poisson bracket structure, it is  not directly clear how to impose single bounds on $a_\infty$ and $b_\infty$ such that the associated series $\bar{Z}_\infty$ converges. In particular, the supremum norm $\|\cdot\|_\infty$ does not entail a bound of the form $\|\{f,g\}\|_\infty\leq M\|f\|_\infty\|g\|_\infty$. Therefore, a similar analysis as done in \cite{Suzuki} seems not possible.  This is exactly the point where the method of majorants comes into play. First, let us formally  write
\begin{align*}
    &\bar{Z}_{\infty,t}:=\lim_{k\to\infty}\bar Z_{\infty,t}^k\quad\mbox{and}\quad \bar Z_{\infty,t}^k:=\sum_{n=1}^k\bar{z}_{\infty,t}(n),
\end{align*}
where $$\bar{z}_{\infty,t}(n) = z(ta_\infty,b_\infty,\{\cdot,\cdot\},n)$$ are the coefficients corresponding to the classical limits of $$\bar{z}_{N,t}(n)=\frac{1}{N} z(tNa_N,Nb_N,[\cdot,\cdot],n),$$ the latter are $t$-dependent and defined as those coefficients appearing in the Baker-Campbell-Hausdorff expansion for $b_N$ and $a_N$ replaced by $ta_N$. Similarly, one may consider also its quantum analog as a formal sum of the $\bar{z}_{N,t}(n)$ terms
\begin{align*}
\bar{Z}_{N,t}=\lim_{k\to\infty}\bar{Z}_{N,t}^k=\lim_{k\to\infty}\sum_{n=1}^k\bar{z}_{N,t}(n),
\end{align*}
with $a_N$ replaced by $ta_N$. As mentioned above, despite the fact that this limit may not exist for finite $N$, at the level of equivalence classes it will converge towards an element of $[\mathfrak{A}]_\infty^\pi$.
\\\\
\noindent
To this avail we proceed in the following manner. As seen in $\S \ref{MFQss and all that}$, the state space $S(\mathfrak{A})\cong \mathcal{Q}_\ell\subset\mathbb{R}^{\ell^2-1}$.
For $f\in C^\infty(\mathcal{Q}_\ell)$, define the following norm
\begin{align*}
    \|f\|^{(k)}:=\sum_{|\alpha|\leq k}\frac{1}{\alpha!}\|\partial^\alpha f\|_\infty<\infty,
\end{align*}
where $\alpha$ is a multi index and $\|\cdot\|_\infty$ is the supremum norm on $C(\mathcal{Q}_\ell)$. By direct inspection, this norm is sub-multiplicative. Indeed,
$$\partial^\alpha(fg)=\sum_{\beta\leq \alpha}\binom{\alpha}{\beta}(\partial^\beta f)(\partial^{\alpha-\beta}g),$$
implying that
\begin{align*}
\|fg\|^{(k)}&=\sum_{|\alpha|\leq k}\frac{1}{\alpha!}\|\partial^\alpha(fg)\|_\infty\\&\leq\sum_{|\alpha|\leq k}\sum_{\beta\leq\alpha}\frac{1}{\beta!(\alpha-\beta)!}\|\partial^\beta f\|_\infty\|\partial^{\alpha-\beta}g\|_\infty\\&\leq \|f\|^{(k)}\|g\|^{(k)}.
\end{align*}
If $\lim_{k\to\infty}\|f\|^{(k)}$ exists, we also define
\begin{align*}
    \|f\|:=\lim_{k\to\infty}\|f\|^{(k)}.
\end{align*}
It is not difficult to see  that $\|\cdot\|$ defines a sub-multiplicative norm as well, and moreover, $\|\cdot \|_\infty\leq\|\cdot\|$.%
\\\\
\noindent
The following result is crucial for the forthcoming analysis.
\begin{lemma}\label{normssss}
Let $f,g\in C^\infty(\mathcal{Q}_\ell)$ be such that both $\|f\|$ and $\|g\|$ are finite. Then, 
    \begin{align*}
        \|\{f,g\}\|\leq M\|f\|\cdot \|g\|,
    \end{align*}
    for some $M>0$ that does not depend on $f$ and $g$. Furthermore, all nested $n$-fold Poisson brackets of $f$ and $g$ have finite norm for each $n\in\mathbb{N}$.
\end{lemma}
\begin{proof}
   By definition, for all $x\in \mathcal{Q}_\ell$
   $$\{f,g\}( x)=\sum_{a,b,c=1}^{\ell^2-1}C_{ab}^cx_c\frac{\partial \tilde f(x)}{\partial x_a}\frac{\partial \tilde g(x)}{\partial x_b},$$
where $\tilde f,\tilde g\in C^\infty(\mathbb{R}^{\ell^2-1})$ are arbitrary extensions of $f,g$, \textit{cf.} \eqref{Poisson}.
Since the norm $\|\cdot\|$ is sub-multiplicative, it follows that
$$\|\{f,g\}\|=\|\{\tilde f,\tilde g\}_{|_{\mathcal{Q}_\ell}}\|\leq\sum_{a,b,c}^{\ell^2-1}C_{ab}^c\|x_c\|\cdot\| f\|\cdot\| g\|=M\cdot\|f\|\cdot\|g\|,$$
since $\|x_c\|_\infty=1$, which follows from the fact that $\mathcal{Q}_\ell\cong S(M_\ell(\mathbb{C}))$, in particular $x_c=\omega(t_c)$, see \eqref{coordinates}, which is bounded by one, since $t_c$ is a   normalized generator
and the derivative of $x_c$ equals one, so that $\|x_c\|=2.$
Here, the constant $M>0$ is defined as
\begin{align}\label{M}
    M:=2(\ell^2-1)^3\max_{1\leq a,b,c\leq \ell^2-1}C_{ab}^c<\infty.
\end{align}
The second statement about nested Poisson brackets follows by induction.
\end{proof}

\begin{remark}
    Recall that a smooth function $w:U\to\mathbb{R}$ is real analytic if and only if for every compact set $K\subset U$ there is $C>0$ such that for every multi-index $\alpha$ the following bound holds
\begin{align*}\label{characterization}     \sup_{{ x}\in K}|\partial^\alpha w({ x})|\leq C^{|\alpha|+1}\alpha!,\end{align*}
see \cite{Ko60}. Now, finiteness of the norm $\|\cdot\|$ precisely gives this estimate for the partial derivatives. Although $S(\mathfrak{A})$ is not an open set, we can parameterize $S(\mathfrak{A})$  in terms of $\mathcal{Q}_\ell$, whose interior is an open subset of $\mathbb{R}^{\ell^2-1}$ \textit{cf.} \eqref{coordinates}. Therefore, one can still use the characterization above for analytic functions by looking at the interior of $S(\mathfrak{A})$ via $\mathcal{Q}_\ell$, where the parameterization is analytic.
\hfill$\blacksquare$
\end{remark}

Given $f,g\in C^\infty(\mathcal{Q}_\ell)$ with $\|f\|,\|g\|<\infty$, for $k\in\mathbb{N}$, we introduce the following sequence of subspace
$$\mathfrak{D}_k(f,g):=\text{span}\{\text{$n$-fold Poisson brackets involving $f$ and $g$}\ | \ 1\leq n\leq k\}.$$
with the convention that $\mathfrak{D}_0(f,g)=\text{span}\{f,g\}$, i.e. the linear span of the functions $f$ and $g$. Now, Lemma~\ref{normssss} yields  by induction that for any $w\in \mathfrak{D}_k(f,g)$, we have $\|w\|<\infty$ for any $k$. 
Since $\mathfrak{D}_k(f,g)\subset \mathfrak{D}_{k+1}(f,g)$, we can consider the completion $$\mathfrak{D}(f,g):= \overline{\bigcup_{k=0}^\infty\mathfrak{D}_k(f,g)}^{\|\cdot\|}$$
within $C^\infty(S(\mathfrak{A}))$.
 Moreover, $\mathfrak{D}(f,g)\neq\{0\}$ as soon as $f$ and $g$ are non-zero.
\\\\
\noindent
From the previous results it follows that
if $f,g$ is such that both $\|f\|$ and $\|g\|$ are finite, the  ensuing sequence defined by $\bar Z_{\infty}^k=\sum_{n=1}^k\bar z_{\infty}(n)$ forms a sequence in $\mathfrak{D}(f,g)$. 
\\\\
\noindent
For $\delta>0$ defined by Theorem~\ref{varr} and $M$ by \eqref{M}, consider now the open set
\begin{align}
 \mathfrak{U}=\{ f \in C^\infty(\mathcal{Q}_\ell)  \ |\ \|f\|< \frac{\delta}{2M} \}.
\end{align}

\begin{proposition}\label{Prop: fundvarvar}
Let $f,g\in C^\infty(\mathcal{Q}_\ell)$ be polynomials in $\mathfrak{U}$ and consider the associated vector space $\mathfrak{D}(f,g)$. 
Then, the series $\bar Z_{\infty}=\sum_{n=1}^\infty\bar z_\infty(n)$ converges absolutely  with respect to $\|\cdot\|$, that is, $\bar Z_{\infty}\in\mathfrak{D}(f,g)$.
\end{proposition}
\begin{proof}
The assertion is an adaptation of \cite[Lemma~2.15.3, Theorem~2.15.4]{Varadarajan} and a repeating application of Lemma \ref{normssss}.
Note that by construction, for each $n\in\mathbb{N}$, the nested Poisson brackets $\bar z_\infty(n)$ satisfy the recurrence relation \eqref{expansion}--\eqref{expansion2}, with $[\cdot,\cdot]$ replaced by $\{\cdot,\cdot\}$. For $n=1$, we estimate $\|\bar z_\infty(1)\|\leq 2\alpha$, where $\alpha=\max{\{\|f\|,\|g\|\}}$ and for $n\geq 2$, the equations relation \eqref{expansion}--\eqref{expansion2},  imply the estimate
\begin{align}\label{estimatess}
&(n+1)\|\bar z_\infty(n+1)\|\\
&\leq \alpha M\|\bar z_\infty(n)\|
+2\alpha\sum_{\substack{p\geq 1,\\ 2p\leq n}} \frac{|B_{2p}|}{(2p)!} \sum_{\substack{k_1,\cdots ,k_{2p}>0\\ k_1+\cdots k_{2p}=n}}M^{2p}\|\bar z_\infty(k_1)\| \cdots \|\bar z_\infty(k_{2p})\|\nonumber.
\end{align}
\noindent
Let $y$ be the solution of the differential equation \eqref{ode}, $$\frac{dy}{dz}=\frac{1}{2}y+\sum_{p=1}^\infty\frac{|B_{2p}|}{(2p)!}y^{2p}$$ 
which is analytic in $\{z \ | \ |z|< \delta \}$ by Theorem~\ref{varr}, and write
\begin{align}\label{sol01}
y(z)=\sum_{n\geq 1}\gamma(n) z^n.
\end{align}
Substituting \eqref{sol01} into \eqref{ode} gives
\begin{align}
 (n+1)\gamma(n+1)&=\frac{1}{2}\gamma(n)+\sum_{\substack{p\geq 1,\\ 2p\leq n}} \frac{|B_{2p}|}{(2p)!} \sum_{\substack{k_1,\cdots ,k_{2p}>0\\ k_1+\cdots k_{2p}=n}}\gamma(k_1)\cdots \gamma(k_{2p}),\label{gamma}\\
\gamma(1)&=1.\nonumber
\end{align}
Then, $\gamma(m)\geq 0$, for all $m$.  We now claim the following
\begin{align}\label{ind}
\|\bar z_\infty(m)\|\leq M^{m-1}(2\alpha)^m\gamma(m).
\end{align}
Indeed, since $\|z_\infty(1)\|\leq 2\alpha$, this is true for $m=1$. Suppose \eqref{ind} is true for all $1\leq n\leq m$. Then, from  \eqref{estimatess} and \eqref{gamma} we get that
\begin{align*}
&(m+1)\|\bar z_\infty(m+1)\|\\&\leq M^{m}(2\alpha)^{m+1}\frac{\gamma(m)}{2}+2\alpha\sum_{\substack{p\geq 1,\\ 2p\leq n}} \frac{ |B_{2p}|}{(2p)!} \sum_{\substack{k_1,\cdots ,k_{2p}>0\\ k_1+\cdots k_{2p}=m}}M^{m}(2\alpha)^{m}\gamma(k_1) \cdots \gamma(k_{2p})\nonumber
\\&=M^m(2\alpha)^{m+1}(m+1)\gamma(m+1).
\end{align*}
Thus, \eqref{ind} is true for all $m\geq 1$, completing the proof of the claim.
\smallskip

\noindent
Since the series \eqref{sol01} converges absolutely if $|z|<\delta$, the previous claim implies that $\sum_{n\geq 1}\bar z_\infty(n)$ converges absolutely whenever $2M\max{\{\|f\|,\|g\|\}}= 2M\alpha< \delta$, i.e. $f,g\in\mathfrak{U}$. As a result,  $\bar{Z}_\infty\in{\mathfrak{D}}(f,g)$.
\end{proof}

\subsection{Large deviation upper bounds}
Finally, we are in a position to prove the convergence of the logarithmic generating function, at least for small $t$, that is, at sufficiently large temperature. As indicated in  $\S$\ref{ldp in statmech} we assume the following set-up. Let $\Lambda,\Lambda'\subset X$ with $M=|\Lambda|$, $M'=|\Lambda'|$ and $N=|X|$. Consider $a'\in \mathfrak{A}^\Lambda$ and $h'\in\mathfrak{A}^{\Lambda'}$, and for $N\geq M,M'$ associate the symmetric sequences $a_N=\pi_N^M(a')$ and  $h_N=\pi_N^{M'}(h')$, as defined in $\S$\ref{MFQss and all that}. In Theorem \ref{prop: moment generating function 21} we provide sufficient conditions such that the function $F_N(t)$ defined by \eqref{momgen1}, i.e.,
\begin{align*}
F_N(t):=\frac{1}{N}\log{\frac{Tr[e^{Nh_N}e^{Nta_N}]}{Tr[e^{Nh_N}]}}=\frac{1}{N}\log{Tr[e^{N\bar Z_{N,t}}]}-\frac{1}{N}\log{Tr[e^{N h_N}]},
\end{align*}
has a limit as $N\to\infty$. We moreover characterize the pertinent limit. To this avail, we prove that $\bar{Z}_{N,t}$ defined by \eqref{section} induces a continuous cross-section of the symmetric bundle. For such sections, a recent result by van de Ven shows that the limit of $F_N$ indeed exists.

\begin{theorem}[Theorem~4.4 in \cite{Ven20222}]\label{vdv}
Let $(a_N)_N$ be a continuous cross-section of symmetric $C^*-$bundle with classical limit $a_\infty$. Then, the mean-field free energy satisfies
\begin{align*}
    \lim_{N\to\infty }\frac{1}{N}\log{Tr[e^{Na_N}]}=\sup_{\omega\in S(\mathfrak{A})}(a_\infty(\omega)-s_\infty(\delta_\omega)),
\end{align*}
where the function $s_\infty$ is the {\em mean-field entropy} defined for any  state $\omega\in S([\mathfrak{A}]_\infty^\pi)$ by
\begin{align}\label{def entropy}
 s_\infty(\omega):=-\lim_{N\to\infty}\frac{1}{N}Tr[\rho_N^\omega\log{\rho_N^\omega}].
\end{align}
Here $\rho_N^\omega$ is the density matrix associated with the state $\omega_N$, the latter given by restriction of $\omega$ to $S(\mathfrak{A}^N)$. For $\omega\in S(\mathfrak{A})$, the state $\delta_\omega$ is an extremal state on $C(S(\mathfrak{A}))\cong [\mathfrak{A}]_\infty^\pi$ given by point evaluation, i.e., $\delta_\omega(f)=f(\omega)$. In fact, all extremal states arising in this way.
\end{theorem}
\begin{remark} 
The mean-field entropy always exists for states on $[\mathfrak{A}]_\infty^\pi$ and defines a weak-$*$ continuous functional on $S([\mathfrak{A}]_\infty^\pi)$.
\hfill$\blacksquare$
\end{remark}

\begin{theorem}\label{prop: moment generating function 21}
Assume the principal symbol $h_\infty$  of the continuous cross section $(h_N)_N$ is in $\mathfrak{U}$.
Then, for $t$ sufficiently small  the function $F_N(t)$ defined by \eqref{momgen1} 
converges as $N\to\infty$, and the limit is given by
$$F(t)=\lim_{N\to\infty}F_N(t)=\sup_{\omega\in S(\mathfrak{A})}(\bar{Z}_{\infty,t}-s_\infty(\omega))-\sup_{\omega\in S(\mathfrak{A})}(h_\infty(\omega)-s_\infty(\omega)),$$
where $s_\infty(\omega)$ is the mean-field entropy, \textit{cf.} Theorem~\ref{vdv}.
\end{theorem}
\begin{proof}
To prove this we first observe that the classical limits $a_\infty$ and $h_\infty$ of the continuous cross section $a_N$ and $h_N$ are polynomials, i.e. the vector space $\mathfrak{D}(a_\infty,h_\infty)\neq\{0\}$. We can therefore find $t>0$  small enough such that $2tM\|a_\infty\|<\delta$, i.e., $ta_\infty\in\mathfrak{U}$. Moreover, the ensuing coefficients $\bar{z}_{\infty,t}(n)$ satisfy the recurrence relation \eqref{expansion}--\eqref{expansion2} with $[\cdot,\cdot]$ replaced by $\{\cdot,\cdot\}$.
Therefore, on account of Proposition~\ref{Prop: fundvarvar} the series $\bar{Z}_{\infty,t}=\sum_{n=1}^\infty\bar{z}_{\infty,t}(n)$ exists in $\mathfrak{D}(a_\infty,h_\infty)$. In particular, the partial sums $ Z_{\infty,t}^k=\sum_{n=1}^k\bar{z}_{\infty,t}(n)$ form a Cauchy sequence, so that
\begin{align*}\label{eqs}
\lim_{k,l\to\infty}\|[\bar{Z}_{N,t}^l]_N- [\bar{Z}_{N,t}^k]_N\|&=\lim_{k,l\to\infty}\lim_{N\to\infty}\|\bar{Z}_{N,t}^l - \bar{Z}_{N,t}^k\|_N\nonumber\\&=\lim_{k,l\to\infty}\|\bar{Z}_{\infty,t}^l- \bar{Z}_{\infty,t}^k\|_\infty\nonumber\\&\leq \lim_{k,l\to\infty}\|\bar{Z}_{\infty,t}^l - \bar{Z}_{\infty,t}^k\|=0,
\end{align*}
where in the second step we have used that the sections $[\bar{Z}_{N,t}^k]_N$ are symmetric sequences and hence continuous, \textit{cf.} Lemma \ref{Lemm: fundamental1}. Hence, $([\bar{Z}_{N,t}^k]_N)_k$ is a Cauchy sequence of elements $[\bar{Z}_{N,t}^k]\in [\mathfrak{A}]_\infty^\pi$. Since $[\mathfrak{A}]_\infty^\pi$ is a $C^*-$algebra, the ensuing limit $[\bar{Z}_{N,t}]_N\in [\mathfrak{A}]_\infty^\pi$.
In particular, the sequence
\begin{align*}
A_{N,t}:=\begin{cases}
    \bar{Z}_{N,t}, \ \ &N\in\mathbb{N};\\
    [\bar{Z}_{N,t}]_N, \ \ &N=\infty;
\end{cases}
\end{align*}
induces a continuous cross-section of the symmetric bundle. Indeed, on account of  condition \eqref{Item: norm continuity for bundle of Cstar algebras} of the definition of a continuous cross-section (see $\S$\ref{subs:intro}),  the norm function is continuous, since
\begin{align*}
\lim_{N\to\infty}\|A_{N,t}\|_N&=\| [\bar{Z}_{N,t}]_N\|=\|j( [\bar{Z}_{N,t}]_N)\|_\infty\\&= \lim_{k\to\infty}\|j([\bar Z_{N,t}^k]_N)\|_\infty=\lim_{k\to\infty}\|\bar{Z}_{\infty,t}^k\|_\infty=\|\bar{Z}_{\infty,t}\|_\infty,
\end{align*}
using that the map $j:[\mathfrak{A}]_\infty^\pi\to C(S(\mathfrak{A}))$ is an isometric $*$-isomorphism, \textit{cf.} \eqref{map j} in the second step, Lemma~\ref{Lemm: fundamental1} in the final last step, and, on account of Proposition~\ref{Prop: fundvarvar} the estimate
\begin{align*}
\lim_{k\to\infty}|\|\bar{Z}_{\infty,t}\|_\infty-\|\bar{Z}_{\infty,t}^k\|_\infty|
\leq \lim_{k\to\infty}\|\bar{Z}_{\infty,t}-\bar{Z}_{\infty,t}^k\|_\infty
\leq \lim_{k\to\infty}\|\bar{Z}_{\infty,t}-\bar{Z}_{\infty,t}^k\|
=0,
\end{align*}
in the last step.
\noindent
Since, $$F_N(t)=\frac{1}{N}\log{Tr[e^{N\bar{Z}_{N,t}}]}-\frac{1}{N}\log{Tr[e^{Nh_N}]},$$
and we now know that $[\bar{Z}_{N,t}]_N\in [\mathfrak{A}]_\infty^\pi$, 
we may apply \cite[Theorem 4.4]{Ven20222}, stating that the  mean-field limit of the free energy for macroscopic observables exists, and is given by
$$\lim_{N\to\infty}F_N(t)=\sup_{\omega\in S(\mathfrak{A})}(\bar{Z}_{\infty,t}(\omega)-s_\infty(\omega))-\sup_{\omega\in S(\mathfrak{A})}(h_\infty(\omega)-s_\infty(\omega)),$$
where $s_\infty(\omega)$ is the mean-field entropy, 
\textit{cf.} \eqref{def entropy}. This concludes the proof of the theorem.
\end{proof}
\begin{remark}
    The previous theorem shows that the limit $F(t)$ is not characterized by substituting $a_N$ and $h_N$ by their classical counterparts, $a_\infty$ and $h_\infty$, but that the interactions between $a_\infty$ and $h_\infty$ also play a role in defining this limit. This indeed confirms that Golden-Thompson inequality, \textit{cf.} \eqref{GT}, remains strict.
    \hfill$\blacksquare$
\end{remark}

Finally, it may be clear that differentiability is in general not guaranteed. However, if one considers for example quantum  averages, i.e. $a_N=t_\mu^{\Lambda}$ and $h_N=t_{\mu'}^{\Lambda}$ where $t_\mu,t_\mu'$ denote two basis elements of the lie algebra of $\mathfrak{s}\mathfrak{u}(\ell)$, \textit{cf.} \eqref{basisliealgebra}, then 
\begin{align*}
    F(t)=\log{Tr[e^{tt_\mu}e^{t_{\mu'}}]}-\log{Tr[e^{t_{\mu'}}]},
\end{align*}
which is differentiable in all real $t$.

\subsection*{Acknowledgements}
The authors thank Teun van Nuland, Nicolo' Drago and Lorenzo Pettinari for their detailed feedback and fruitful discussions. The work of Christiaan J.F. van de Ven has been supported by  the Deutsche Forschungsgemeinschaft (DFG, German Research Foundation)– 470903074; 465199066. Matthias Keller also acknowledges the financial support of the DFG.

\end{document}